\documentclass[a4paper,11pt]{article}
\pdfoutput=1

\usepackage[vmargin=33mm,hmargin=33mm]{geometry}
\usepackage[UKenglish]{babel}
\usepackage[utf8]{inputenc}
\usepackage{microtype}
\usepackage{amsfonts,amsmath,amssymb,amsthm,booktabs,cite,color,doi,enumitem,graphicx,url}
\usepackage{hyperref}
\usepackage{thmtools}
\usepackage{thm-restate}

\newtheorem{theorem}{Theorem}
\newtheorem{lemma}[theorem]{Lemma}

\theoremstyle{remark}
\newtheorem{remark}{Remark}
 
\newcommand{\cS}{\mathcal{S}}
\newcommand{\logst}{\log^{*}}
\newcommand{\pt}[1]{{}^{#1}2} 

\hypersetup{
    colorlinks=true,
    linkcolor=black,
    citecolor=black,
    filecolor=black,
    urlcolor=black,
}

\newenvironment{mycover}
               {\list{}{\listparindent 0pt
                        \itemindent    \listparindent
                        \leftmargin    0pt
                        \rightmargin   0pt
                        \parsep        0pt}%
                \raggedright
                \item\relax}
               {\endlist}

\begin{document}

\hypersetup{
    pdfauthor={Joel Rybicki, Jukka Suomela},
    pdftitle={Exact bounds for distributed graph colouring},
}

\begin{mycover}
{\LARGE \textbf{Exact bounds for distributed graph colouring}\par}

\bigskip

\medskip
\textbf{Joel Rybicki}\\

\smallskip
{\small Helsinki Institute for Information Technology HIIT, \\
Department of Computer Science, Aalto University

\medskip
Department of Algorithms and Complexity, \\
Max Planck Institute for Informatics

\medskip
\nolinkurl{joel.rybicki@aalto.fi}\par}

\bigskip
\textbf{Jukka Suomela}\\

\smallskip
{\small Helsinki Institute for Information Technology HIIT, \\
Department of Computer Science, Aalto University

\medskip
\nolinkurl{jukka.suomela@aalto.fi}\par}

\end{mycover}

\bigskip

\paragraph{Abstract.}

We prove exact bounds on the time complexity of distributed graph colouring. If we are given a directed path that is properly coloured with $n$ colours, by prior work it is known that we can find a proper $3$-colouring in $\frac{1}{2} \logst(n) \pm O(1)$ communication rounds. We close the gap between upper and lower bounds: we show that for infinitely many $n$ the time complexity is precisely $\frac{1}{2} \logst n$ communication rounds.

\thispagestyle{empty}
\setcounter{page}{0}
\newpage


\section{Introduction}

One of the key primitives in the area of distributed graph algorithms is \emph{graph colouring in directed paths}. This is a fundamental symmetry-breaking task, widely studied since the 1980s---it is used as a subroutine in numerous efficient distributed algorithms, and it also serves as a convenient starting point in many lower-bound proofs. In the 1990s it was already established that the distributed computational complexity of this problem is $\frac{1}{2} \logst(n) \pm O(1)$ communication rounds \cite{cole86deterministic,linial92locality}. We are now able to give \emph{exact} bounds on the distributed time complexity of this problem, and the answer turns out to take a surprisingly elegant form:

\begin{restatable}{theorem}{mainthm}\label{thm:main}
    For infinitely many values of $n$, it takes exactly $\frac{1}{2}\logst n$ rounds to compute a $3$-colouring of a directed path.
\end{restatable}

\subsection{Problem Setting}

Throughout this work we focus on \emph{deterministic} distributed algorithms. As is common in this context, what actually matters is not the number of nodes but the range of their labels. For the sake of concreteness, we study precisely the following problem setting:
\begin{quote}
    We have a path or a cycle with any number of nodes, and the nodes are properly coloured with colours from $[n] = \{1,2,\dotsc,n\}$.
\end{quote}
The techniques that we present in this work can also be used to analyse other variants of the problem---for example, a cycle with $n$ nodes that are labelled with some permutation of $[n]$, or a path with at most $n$ nodes that are labelled with unique identifiers from $[n]$. However, the exact bounds on the time complexity will slightly depend on such details.

We will assume that there is a globally consistent orientation in the path: each node has at most one predecessor and at most one successor. Our task is to find a proper colouring of the path with $c$ colours, for some number $c \ge 3$. We will call this task \emph{colour reduction from $n$ to $c$}.

We will use the following model of distributed computing. Each node of the graph is a computational entity. Initially, each node knows the global parameters $n$ and $c$, its own label from $[n]$, its degree, and the orientations of its incident edges. Computation takes place in synchronous communication rounds. In each round, each node can send a message to each of its neighbours, receive a message from each of its neighbours, update its state, and possibly stop and output its colour. The \emph{running time} of an algorithm is defined to be the number of communication rounds until all nodes have stopped. We will use the following notation:
\begin{itemize}
    \item $C(n,c)$ is the time complexity of colour reduction from $n$ to $c$.
    \item $T(n,c)$ is the time complexity of colour reduction from $n$ to $c$ if we restrict the algorithm so that a node can only send messages to its successor. We call such an algorithm \emph{one-sided}, while unrestricted algorithms are \emph{two-sided}.
\end{itemize}
We can compose colour reduction algorithms, yielding $C(a,c) \le C(a,b) + C(b,c)$ and $T(a,c) \le T(a,b) + T(b,c)$ for any $a \ge b \ge c$. It is easy to see (shown in Lemma~\ref{lemma:hardness}) that
\[
    C(n,c) = \lceil T(n,c)/2 \rceil.
\]
We will be interested primarily in $C(n,c)$, but function $T(n,c)$ is much more convenient to analyse when we prove upper and lower bounds.

\subsection{Prior Work}

The asymptotically optimal bounds of
\[
    \logst(n) - O(1) \le T(n,3) \le \logst(n) + O(1)
\]
are covered in numerous textbooks and courses on distributed and parallel computing \cite{cormen90introduction,peleg00distributed,barenboim13distributed,wattenhofer13lecture,suomela-dabook}. The proof is almost unanimously based on the following classical results:
\begin{description}
    \item[Cole--Vishkin colour reduction (CV):] The upper bound was presented in the modern form by Goldberg, Plotkin, and Shannon \cite{goldberg88parallel} and it is based on the technique first introduced by Cole and Vishkin \cite{cole86deterministic}. The key ingredients are a fast colour reduction algorithm that shows that $T(2^k,2k) \le 1$ for any $k \ge 3$, and a slow colour reduction algorithm that show that $T(k+1,k) \le 2$ for any $k \ge 3$. By iterating the fast colour reduction algorithm, we can reduce the number of colours from $n$ to $6$ in $\logst(n) \pm O(1)$ rounds, and by iterating the slow colour reduction algorithm, we can reduce the number of colours from $6$ to $3$ in $6$ rounds (with one-sided algorithms).
    \item[Linial's lower bound:] The lower bound is the seminal result by Linial \cite{linial92locality}. The key ingredient is a speed-up lemma that shows that $T(n,2^c) \le T(n,c) - 1$ when $T(n,c) \ge 1$. By iterating the speed-up lemma for $\logst(n)-3$ times, we have $T(n,4) \ge T(n,k) + \logst(n) - 3$ for a $k < n$. Clearly $T(n,3) \ge T(n,4)$ and $T(n,k) \ge 1$, and hence $T(n,3) \ge \logst(n) - 2$.
\end{description}

In the upper bound, many sources---including the original papers by Cole and Vishkin and Goldberg et al.---are happy with the asymptotic bounds of $\logst(n) + O(1)$ or $O(\logst n)$. However, there are some sources that provide a more careful analysis. The analysis by Barenboim and Elkin \cite{barenboim13distributed} yields
$
    T(n,3) \le \logst(n) + 9,
$
and the analysis in the textbook by Cormen et al.~\cite{cormen90introduction} yields
$
    T(n,3) \le \logst(n) + 7
$.
In our lecture course \cite{suomela-dabook} we had an exercise that shows how to push it down to
\[
    T(n,3) \le \logst(n) + 6.
\]

In the lower bound, there is less variation. Linial's original proof \cite{linial92locality} yields $T(n,3) \ge \logst(n) - 3$, and many sources \cite{barenboim13distributed,suomela-dabook,laurinharju14linial-easy} prove a bound of
\[
    T(n,3) \ge \logst(n) - 2.
\]

On the side of lower bounds, nothing stronger than Linial's result is known. There are alternative proofs based on Ramsey's theorem \cite{czygrinow08fast} that yield the same asymptotic bound of $T(n,3) = \Omega(\logst n)$, but the constants one gets this way are worse than in Linial's proof.

On the side of upper bounds, however, there is an algorithm that is strictly better than CV: \textbf{Naor--Stockmeyer colour reduction (NS)} \cite{naor95what}. While CV yields $T(2^k,2k) \le 1$ for any $k \ge 3$, NS yields a strictly stronger claim of $T(\binom{2k}{k},2k) \le 1$ for any $k \ge 2$. However, the exact bounds that we get from NS are apparently not analysed anywhere, and their algorithm is hardly ever mentioned in the literature. Hence the state of the art appears to be
\begin{align*}
    \logst(n) - 2 &\le T(n,3) \le \logst(n) + 6, \\
    \frac{1}{2}\logst(n) - 1 &\le C(n,3) \le \frac{1}{2}\logst(n) + 3.
\end{align*}
Note that we have $\logst n \le 5$ for all $n < 10^{19728}$, and hence in practice the constant term $6$ dominates the term $\logst n$ in the upper bound.

\subsection{Contributions}

In this work we derive \emph{exact bounds} on $C(n,3)$ for infinitely many values of $n$, and near-tight bounds for all values of $n$. We prove that for infinitely many values of $n$
\begin{align*}
    C(n,3) = \frac{1}{2} \logst n,
\end{align*}
and for all sufficiently large values of $n$
\begin{align*}
    \logst(n) - 1 &\le T(n,3) \le \logst(n) + 1.
\end{align*}
With $C(n,3) = \lceil T(n,3)/2 \rceil$ this gives a near-complete picture of the exact complexity of colouring directed paths. The key new techniques are as follows:
\begin{enumerate}
    \item We give a new analysis of NS colour reduction.
    \item We give a new lower-bound proof that is strictly stronger than Linial's lower bound.
    \item We show that \emph{computational techniques} can be used to prove not only upper bounds but also lower bounds on $T(n,c)$, also for the case of a general $n$ and not just for fixed small values of $n$ and $c$. We introduce \emph{successor graphs} $\cS_i$ that are defined so that a graph colouring of $\cS_i$ with a small number of colours implies an improved bound on $T(n,3)$.
\end{enumerate}
This work focuses on colour reduction, i.e., the setting in which we are given a proper colouring as an input. Our upper bounds naturally apply directly in more restricted problems (e.g., the input labels are unique identifiers). Our lower bounds results do not hold directly, but the key techniques are still applicable: in particular, the successor graph technique can be used also in the case of unique identifiers.

\subsection{Applications}

Graph colouring in paths, and the related problems of graph colouring in rooted trees and directed pseudoforests, are key symmetry-breaking primitives that appear as subroutines in numerous distributed algorithms for various graph problems~\cite{goldberg88parallel,panconesi01some,garay98sublinear,astrand10vc-sc,czygrinow08fast,lenzen14steiner}.

One of the most direct application of our results is related to colouring \emph{trees}: In essence, colour reduction from $n$ to $c$ in trees with \emph{arbitrary} algorithms is the same problem as colour reduction from $n$ to $c$ in paths with \emph{one-sided} algorithms. Informally, in the worst case the children contain all possible coloured subtrees and hence ``looking down'' in the tree is unhelpful, and we can equally well restrict ourselves to ``looking up'' towards the root. Hence our bounds on $T(n,3)$ can be directly interpreted as bounds on colour reduction from $n$ to $3$ in trees.

The bounds have also applications outside distributed computing. A result by Fich and Ramachandran \cite{fich90linkedlists} demonstrates that bounds on $C(n,3)$ have direct implications in the context of \emph{decision trees} and \emph{parallel computing}.

Indeed, the fastest known \emph{parallel} algorithms for colouring linked lists are just adaptations of CV and NS colour reduction algorithms. These algorithms reduce the number of colours very rapidly to a relatively small number (e.g., dozens of colours), and the key bottleneck has been pushing the number of colours down to $3$. In particular, reducing the number of colours down to $3$ with state-of-the-art algorithms has been much more expensive than reducing it to $4$, but this phenomenon has not been understood so far. Prior bounds on $T(n,c)$ have not been able to show that the case of $c = 3$ is necessarily more expensive than $c = 4$. Our improved bounds are strong enough to separate $T(n,4)$ and $T(n,3)$.

From the perspective of practical algorithm engineering and programming, this work shows that we should avoid CV colour reduction, but we can be content with NS colour reduction; the former incurs a significant overhead (e.g., in terms of linear scans over the data in parallel computing), but the latter is near-optimal.

\section{Preliminaries}

\paragraph{Sets and Functions.} For any positive integer $k$, we use $[k]$ to denote the set $\{1, 2, \dots, k\}$. For any set $X$, we use $2^{X} = \{ Y \subseteq X \}$ to denote the powerset of $X$. Define the \emph{iterated logarithm} as
\begin{align*}
 \log^{(0)}(x) &= x, \\
 \log^{(i+1)}(x) &= \log^{(i)}( \log x ) \textrm{ for all } i \ge 0.
\end{align*}
In this work, all logarithms are in base 2. 
Moreover, the \emph{log-star} function is
\[
 \logst x = \min \{ i : \log^{(i)} x \le 1 \}.
\]
Finally, we define the \emph{tetration}, or a power tower, with base 2 as 
\begin{align*}
 \pt{0} &= 1, \\
 \pt{i+1} &= 2^{\left(\pt{i} \right)} \textrm{ for all } i \ge 0.
\end{align*}

\begin{figure}[t]
    \centering
    \includegraphics[page=1]{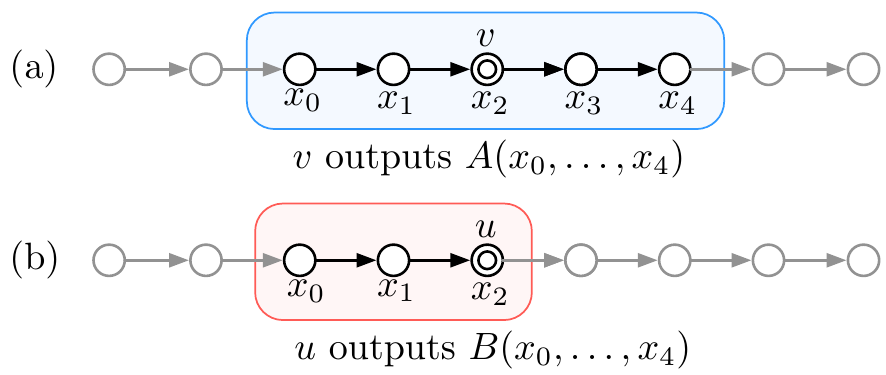}
    \caption{The difference of two-sided and one-sided algorithms. (a) A two-sided algorithm $A$ that runs for 2 rounds. (b) A one-sided algorithm $B$ that runs for 2 rounds. \label{fig:algorithms}}
\end{figure}

\paragraph{Algorithms.}

In this work, we focus on algorithms that run on directed paths. We distinguish between two-sided and one-sided algorithms; see Figure~\ref{fig:algorithms}. \emph{Two-sided algorithms} correspond to the usual notion of an algorithm in the LOCAL model: an algorithm running for $t$ rounds has to decide on its output using the information available at most $t$ hops away. Formally, a two-sided $c$-colouring algorithm corresponds to a function
\[
 A \colon [n]^{2t+1} \to [c].
\]
Moreover, as $A$ outputs a proper colouring, the function satisfies $A(x_0, \dots, x_{2t}) \neq A(x_1, \dots, x_{2t+1})$ when $x_{i} \neq x_{i+1}$ for all $i \ge 0$.

In contrast to two-sided algorithms, \emph{one-sided algorithms} are algorithms in which nodes can only send messages to successors. Therefore, a one-sided algorithm that runs in $t$ rounds can only gather information from at most $t$ \emph{predecessors}. Formally, a one-sided $c$-colouring algorithm $B$ that runs for $t$ steps corresponds to a function
\[
 B \colon [n]^{t+1} \to [c],
\]
which satisfies $B(x_0, \dots, x_{t}) \neq B(x_1, \dots, x_{t+1})$ when $x_{i} \neq x_{i+1}$ for all $i \ge 0$.

\begin{figure}[t]
    \centering
    \includegraphics[page=2]{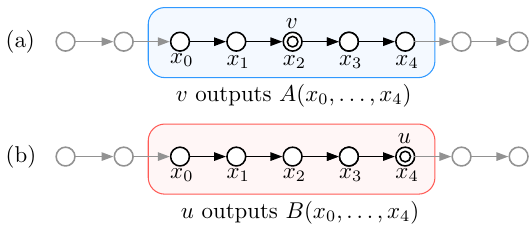}
    \caption{The correspondence between two-sided and one-sided algorithms. (a) A two-sided algorithm $A$ that runs for 2 rounds. (b) A one-sided algorithm that runs in 4 rounds. Both nodes see the same information, so $v$ can easily simulate $B$ and $u$ can simulate $A$. \label{fig:equivalence}}
\end{figure}

It is now easy to see that $C(n,c) = \lceil T(n,c)/2 \rceil$ holds. Intuitively, the connection is straightforward. For example, Figure~\ref{fig:equivalence} illustrates how a $t$-time two-sided algorithm can gather the same information as a $2t$-time one-sided algorithm. For the sake of completeness, we now prove this formally.

\begin{lemma}\label{lemma:hardness}
$C(n,c) = \left \lceil T(n,c)/2 \right \rceil$.
\end{lemma}
\begin{proof}
 First, we show that $T(n,c) \le 2 C(n,c)$. Let $t = C(n,c)$ and $A \colon [n]^{2t+1} \to [c]$ be a two-sided $c$-colouring algorithm that runs in time $t$. We construct a one-sided $c$-colouring algorithm that runs in time $2t$. Recall that a one-sided algorithm can only receive messages from predecessors. Initially, every node sends its own colour to its successor. Then for $2t-1$ rounds we send the colour received from the predecessor to the successor---in the case that a node has no predecessors, the node can simply simulate a properly coloured path preceding it. After $2t$ rounds the node knows its own colour and the colours of its $2t$ predecessors, that is, a vector $(x_0, \dots, x_{2t}) \in [n]^{2t+1}$. Outputting the value $A(x_0, \dots, x_{2t})$ yields a proper colouring. 

 Second, we show that $C(n,c) \le \left\lceil T(n,c) / 2 \right \rceil$. Let $t = \left\lceil T(n,c) / 2 \right \rceil$ and $B \colon [n]^{T(n,c)+1} \to [c]$ a one-sided algorithm that only receives messages from predecessors. Every node sends its colour to both neighbours and then forwards any messages in the $t-1$ subsequent rounds. As $2t \ge T(n,c)$, after $t$ rounds every node knows the colours $(x_0, \dots, x_{T(n,c)})$ in its local neighbourhood. Now the node can output $B(x_0, \dots, x_{T(n,c)})$ which gives a proper colouring.

 Finally, since the time complexity has to be integral---there are no ``half-rounds''---we get that $C(n,c) = \left \lceil T(n,c)/2 \right \rceil$.
\end{proof}

\section{The Upper Bound}

In this section, we bound $T(n,c)$ from above. To do this, we analyse the Naor--Stockmeyer (NS) colour reduction algorithm~\cite{naor95what}. The NS algorithm is one-sided, thus yielding upper bounds for $T(n,c)$.

Let us first recall the NS colour reduction algorithm. Let $n \le { 2k \choose k }$ for some $k \ge 2$ and fix an injection $f \colon [n] \to X$, where $X = \{ Y \subseteq [2k] : |Y| = k \}$. That is, we interpret all colours from $[n]$ as distinct $k$-subsets of $[2k]$.

The algorithm works as follows. First, all nodes send their colour to the successor. Then a node with colour $v$ receiving colour $u$ from its predecessor will output
\[
 A(u,v) = \min f(u) \setminus f(v).
\]
It is easy to show that if $u \neq v \neq w$, then $A(u,v) \in [2k]$ and $A(u,v) \neq A(v,w)$ holds. Thus, $A$ is a one-sided colour reduction algorithm that reduces the number of colours from ${2k \choose k}$ to $2k$ colours in one round and we have that $T\bigl({2k \choose k}, 2k\bigr) = 1$ for any $k \ge 2$.

The above algorithm cannot reduce the number of colours below 4. To reduce the number of colours from four to three, we can use the following one-sided algorithm $B$ that outputs
\[
 B(u,v,w) = \begin{cases}
              \min \{1,2,3\} \setminus \{u,w\} &\text{if } v = 4, \\
              v & \text{otherwise}.
            \end{cases}
\]
The algorithm uses two rounds and this is optimal by Lemma~\ref{lemma:4to3} in Section~\ref{sec:lb}.

We now show the following upper bounds for $T(n,c)$ using the NS colour reduction algorithm.

\begin{lemma}\label{lemma:ns}
 The function $T$ satisfies the following:
 \begin{enumerate}[label=(\alph*)]
  \item $T\bigl(\frac{3}{2} \cdot 2^c,\,\, \frac{3}{2} \cdot c \bigr) = 1$ for any $c = 4h$, where $h > 1$, 
  \item $T\bigl(\frac{3}{2} \cdot \pt{r+4},\,\, \frac{3}{2} \cdot \pt {4} \bigr) \le r$ for any $r \ge 0$,
  \item $T\bigl(\frac{3}{2} \cdot \pt{4},\,\, 3 \bigr) \le 5$.
 \end{enumerate}
\end{lemma}
\begin{proof}\mbox{}
 \begin{enumerate}[label=(\alph*)]
 \item As discussed, the NS colour reduction algorithm shows that $T\bigl( {2k \choose k}, 2k \bigr) = 1$ for $k \ge 2$. Recall the following bound for the central binomial coefficent
\[
 {2k \choose k} \ge \frac{4^k}{\sqrt{4k}}
\]
and let $2k = 3c/2$. Since $c \ge 8$ it follows that 
\[
{2k \choose k } \ge \frac{(2\cdot2)^{3c/4}}{\sqrt{3c}} = \frac{2^{c/2}}{\sqrt{3c}} \cdot 2^c > \frac{3}{2} \cdot 2^c.
\]

 \item To show the claim, it suffices to apply part (a) for $r$ times. 

 \item As ${20 \choose 10} > \frac{3}{2} \cdot \pt{4}$, we can reduce the number of colours to 4 in three rounds as follows: ${20 \choose 10} \leadsto {6 \choose 3} \leadsto {4 \choose 2} \leadsto 4$. By Lemma~\ref{lemma:4to3}, the remaining two rounds can be used to remove the fourth colour. \qedhere
 \end{enumerate}
\end{proof}

\begin{theorem}\label{thm:ub}
 $T(\pt{h},3) \le T(\pt{h}+1,3) \le h + 1$ holds for any $h > 1$.
\end{theorem}
\begin{proof}
The cases $2 \le h \le 4$ follow from the proof of Lemma~\ref{lemma:ns}c. Suppose $h = r + 4$ for some $r > 0$. By Lemma~\ref{lemma:ns}b~and~c we can get a 3-colouring in $r + 5 = h+1$ rounds.
\end{proof}

\section{The Lower Bound}\label{sec:lb}

In this section, we give a new lower bound for the time complexity of one-sided colour reduction algorithms. The proof follows the basic idea of Linial's proof~\cite{linial92locality} adapted to the case of colour reduction, but we show a new lemma that can be used to tighten the bound.

The proof is structured as follows. First, we show that $T(n, 2^c-2) \le T(n,c) - 1$, that is, given a $c$-colouring algorithm, we can devise a faster algorithm that uses at most $2^c-2$ colours; this is just a minor tightening of the usual standard bound, and should be fairly well-known. Second, we prove that a fast 3-colouring algorithm implies a fast 16-colouring algorithm, more precisely, $T(n,16) \le T(n,3)-2$; this is the key contribution of this section. Together these yield the following new bound:

\begin{restatable}{theorem}{lbthm}
\label{thm:lb}
 For any $h > 1$, we have $T(\pt{h}, 3) \ge h$.
\end{restatable}

\subsection{The Speed-up Lemma\label{sec:speedup}}

\begin{lemma}\label{lemma:speedup}
 If $T(n,c) \ge 1$, then $T(n, 2^c-2) \le T(n,c) - 1$.
\end{lemma}
\begin{proof}
 Let $t = T(n,c)$ and $A \colon [n]^{t+1} \to [c]$ be a one-sided $c$-colouring algorithm. We will construct a faster one-sided algorithm $B$ as follows. Consider a node $u$ and its successor $v$. In $t-1$ rounds, node $u$ can find out the colours of its $t-1$ predecessors and its own colour, that is, some vector $(x_0, \dots, x_{t-1}) \in [n]^t$. In particular, node $u$ now knows what information node $v$ can gather in $t$ rounds \emph{except} the colour of $v$ since $A$ is one-sided. However, $u$ can enumerate all the possible outputs of $v$ which give the set
\[
 B(x_0, \dots, x_{t-1}) = \bigl\{ A(x_0, \dots, x_{t-1}, y) \colon y \neq x_{t-1}, y \in [n]\bigr\} \subseteq [c].
\]
Clearly $B(x_0, \dots, x_{t-1}) \ne \emptyset$. We also have $B(x_0, \dots, x_{t-1}) \ne [c]$: For the sake of contradiction, suppose otherwise. This would imply that $v$ could output any value in $[c]$. In particular, if $u$ outputs $A(z, x_0, \dots, x_{t-1}) = a$ for some $z \in [n]$, we could pick $y \in [n]$ such that $A(x_0, \dots, x_{t-1}, y) = a$ as well. However, this would contradict the fact that $A$ was a colouring algorithm. Hence there exists an injection $f$ that maps any possible set $B(\cdot)$ to a value in $[2^c-2]$. 

It remains to argue that no two adjacent nodes construct the same set. Suppose a node $u$ outputs set $X$ and its successor $v$ also outputs $X$. Now we can pick $k \in X$ such that 
\[
 A(x_0, \dots, x_{t-1}, y) = k = A(x_1, \dots, x_{t-1}, y, y')
\]
for some $x_{t-1} \neq y \neq y'$ contradicting that $A$ outputs a proper colouring. Therefore, $f \circ B$ is a one-sided $(2^c-2)$-colouring algorithm that runs in time $t-1 = T(n,c) - 1$.
\end{proof}

\begin{lemma}\label{lemma:rounds}
For any $r>0$, we have $T(\pt{r+3}, 16) \ge r+1$.
\end{lemma}
\begin{proof}
Fix $r>0$. We repeatedly apply Lemma~\ref{lemma:speedup}. Now suppose we have an algorithm that reduces the number of colours from $n$ to $16=\pt{3}$ in $r$ rounds. That is, $T(n, \pt{3}) \le r$ holds for some $n\ge3$. From Lemma~\ref{lemma:speedup} it follows that
\begin{align*}
T(n, \pt{3}) \le r &\implies T(n, \pt{4}-2) \le r - 1 \\
                   &\implies \cdots \\
                   &\implies T(n, \pt{3+r}-2) \le 0,
\end{align*}
but as $T(k,k-1) \ge 1$ for any $k$ it follows that $n < \pt{3+r}$. Thus, $T(\pt{r+3},16) \ge r +1$.
\end{proof}

\subsection{Proof of Theorem~\ref{thm:lb}}

In addition to the speed-up lemma, we need a few more lemmas that bound $T(n,3)$ below for small values of $n$.

\begin{lemma}\label{lemma:4to3}
 $T(4,3) \ge 2$.
\end{lemma}
\begin{proof}
 Let $B' \colon (u,v) \to \{1,2,3\}$ be a one-sided 3-colouring algorithm that runs in one round. Now $B'$ yields a partitioning of the possible input pairs $(u,v)$ where $u \neq v$. It is simple to check that there always exists a pair $(u,v)$ with $u\neq v$ such that there also exists some $w \neq v$ satisfying $B'(u,v) = B'(v,w)$.
\end{proof}

\begin{restatable}{lemma}{compbounds}
 \label{lemma:comp-bounds}
 $T(16,3) \ge 3$.
\end{restatable}
\begin{proof}
As observed by Linial~\cite{linial92locality}, we can show $C(n,c) = t$ if the so-called \emph{neighbourhood graph} $\mathcal{N}_{n,t}$ has a chromatic number of $c$. While Linial analytically bounded the chromatic number of such graphs, we can also compute their chromatic numbers exactly for small values of $n$, $c$, and $t$; see~\cite{rybicki11msc} for a detailed discussion. We use the latter technique to show the claimed bound. That is, the neighbourhood graph $\mathcal{N}_{7,1}$ is not 3-colourable. 

The neighbourhood graph $\mathcal{N}_{7,1} = (V,E)$ is defined as follows. The set of vertices is 
\[
 V = \{ (x_0, x_1, x_2) \in [n]^3 : x_0 \neq x_1 \neq x_2, x_0 \neq x_2 \},
\]
where $n = 7$ and the set of edges is 
\[
 E = \{ \{u,v\} : u,v \in V, u=(x_0, x_1, x_2), v= (x_1, x_2, x_3) \}.
\]
It is easy to check with a computer (e.g.\ using any off-the-shelf SAT or an IP solver) that the graph $\mathcal{N}_{7,1}$ is not 3-colourable. Therefore, $C(7,3) > 1$ and in particular $T(16,3) \ge T(7,3) > 2$.
\end{proof}

To get a lower bound for 3-colouring, we show in the following sections that the existence of a $t$-time one-sided 3-colouring algorithm implies a $(t-2)$-time one-sided 16-colouring algorithm. 

\begin{lemma}
\label{lemma:16-cols}
For any $n \ge 16$, it holds that $T(n,16) \le T(n, 3) - 2$.
\end{lemma}

Now we have all the results for showing the lower bound.

\lbthm*
\begin{proof}
 The cases $r = 2$ and $r = 3$ follow from Lemmas~\ref{lemma:4to3}~and~\ref{lemma:comp-bounds}. For the remaining cases, let $h = r+3$ for some $r > 0$. Suppose $T(\pt{h}, 3) = T(\pt{r+3}, 3) < h$. Then by Lemma~\ref{lemma:16-cols} we would get that $T(\pt{r+3}, 16) < h-2 = r+1$ which contradicts Lemma~\ref{lemma:rounds}.
\end{proof}

\subsection{Proof of Lemma~\ref{lemma:16-cols} via Successor Graphs}

To prove Lemma~\ref{lemma:16-cols}, we analyse the chromatic number of so-called \emph{successor graphs}---a notion similar to Linial's neighbourhood graphs~\cite{linial92locality}. In the following, given a binary relation $R$, we will write $x \in R(y)$ to mean $(y,x) \in R$.

\paragraph{Colouring Relations.}

Suppose $A = A_0$ is a one-sided $3$-colouring algorithm that runs in $t$ rounds. Let $A_1, \dots, A_t$ denote the one-sided algorithms given by iterating Lemma~\ref{lemma:speedup} and $C_{k+1} \subseteq 2^{C_k}$ be the set of colours output by algorithm $A_{k+1}$.

In the following, let $t' = t-k$. Define the \emph{potential successor relation} $S_k \subseteq C_k \times C_k$ to be a binary relation such that $(x,y) \in S_k$ if there exists $x_0, \dots, x_{t'}$ where $x_{i} \neq x_{i+1}$ such that 
\[
 A_k(x_0, \dots, x_{t'-1}) = x \text{ and } A_k(x_1, \dots, x_{t'}) = y.
\]
That is, in the output of algorithm $A_k$ there can be an $x$-coloured node with a successor of colour $y$. Moreover, define the \emph{output relation} $R_k \subseteq C_k \times C_{k+1}$ such that $(x,X) \in R_k$ if 
\[
 A_k(x_0, \dots, x_{t'-2}, x) = X
\]
for some $x_0, \dots, x_{t'-2}$ where $x_{i} \neq x_{i+1}$. That is, a node with colour $x$ can output colour $X$ when executing $A_{k+1}$. From the construction of $A_{k+1}$ given in Lemma~\ref{lemma:speedup}, we get that $R_k = \{ (x,X) : X \subseteq S_k(x), X \neq \emptyset \}$.

\begin{lemma}
 Suppose $X \in R_k(x)$, $Y \in R_k(y)$, and $y \in X$ for some $x,y \in C_k$, then $(X,Y) \in S_{k+1}$ holds. Moreover, the converse holds.
\end{lemma}
\begin{proof}
 As we have $y \in X \subseteq S_k(x)$, this means that a node with colour $x$ may have a successor of colour $y$ after executing algorithm $A_k$. Moreover, as $X \in R_k(x)$ and $Y \in R_k(y)$ hold, then a node with colour $x$ may output $X$ and node with colour $y$ may output $Y$ when executing $A_{k+1}$. Thus, after executing $A_{k+1}$ we may have a node with colour $X$ that has a successor with colour $Y$. Therefore, $(X,Y) \in S_{k+1}$.

 To show the converse, suppose that $(X,Y) \in S_{k+1}$, that is, in some output of $A_{k+1}$ a node $u$ with colour $X$ having a successor $v$ with colour $Y$. Now there must exist some colour $x$ that $X \in R_k(x)$ and some colour $y$ such that $Y \in R_k(y)$. As $v$ is a successor of $u$, the algorithm $A_{k+1}$ outputs a set $X$ consisting of all possible colours for any successor of $u$, and thus, we have $y \in X$.
\end{proof}

\paragraph{Successor Graphs.} For any choice of $A=A_0$, we can construct the successor relation $S_k$ and using this relation, we can define the \emph{successor graph} of $A$ to be the graph $\cS_k(A) = (C_k, E_k)$, where $E_k = \{ \{x,y\} : (x,y) \in S_k \}$. These graphs have the following property:
\begin{lemma} \label{lemma:chromatic}
 Let $\cS_k = (C_k, S_k)$ be the successor graph of $A$, and let $t$ be the running time of $A$. If $f \colon C_k \to [\chi]$ is a proper colouring of $\cS_k$, then $f \circ A_k$ is a one-sided $\chi$-colouring algorithm that runs in $t-k$ rounds. That is, $T(n,\chi) \le t-k$.
\end{lemma}
\begin{proof}
 Let $u$ be the predecessor of $v$ on a directed path. Now by definition, 
\begin{align*}
         &A_k(x_0, \dots, x_{t-1}, u) = x \neq y = A_k(x_1, \dots, x_{t-1}, u, v) \\
&\implies (x,y) \in S_k \implies f(x) \neq f(y).
\end{align*}
Therefore, $f \circ A_k$ is a one-sided $\chi$-colouring algorithm.
 \end{proof}

In the next section, we show the following lemma from which Lemma~\ref{lemma:16-cols} follows.

\begin{lemma} \label{lemma:s2-16-cols}
 For any $t$-time 3-colouring algorithm $A$, the successor graph $\cS_2(A)$ can be coloured with $16$ colours.
\end{lemma}
In particular, this holds for an optimal algorithm $A$ with a running time of $t = T(n,3)$. Together with Lemma~\ref{lemma:chromatic}, this implies Lemma~\ref{lemma:16-cols}. We next show how to prove Lemma~\ref{lemma:s2-16-cols} in two ways: with computers, and without them.

\subsection{A Human-Readable Proof of Lemma~\ref{lemma:s2-16-cols}}

We start by giving a traditional human-readable proof for Lemma~\ref{lemma:s2-16-cols}. That is, we argue that for any one-sided 3-colouring algorithm $A=A_0$ the successor graph $\cS_2(A)$ can be coloured with $16$ colours. Later in Section~\ref{sec:computer-proof}, we give a computational proof of the same result. In the following, we fix $A$ and denote $\cS_2 = \cS_2(A)$ for brevity. 

\paragraph{Structural Properties.} 

We start with the following observations.
\begin{remark}\label{remark:colors}
Sets $C_0$ and $C_1$ satisfy
\begin{align*}
    C_0 &\subseteq \{1,2,3\}, \\
    C_1 &\subseteq \bigl\{\{1\},\{2\},\{3\},\{1,2\},\{1,3\},\{2,3\}\bigr\}.
\end{align*}
\end{remark}

\begin{remark}\label{remark:successors}
Relation $S_1$ satisfies
\begin{align*}
    S_1(i) &\subseteq \big\{ X \in C_1 : i \notin X \big\}, \\
    S_1(\{i,j\}) &\subseteq \big\{ X \in C_1 : \{i,j\} \nsubseteq X \big\}.
\end{align*}
\end{remark}

\begin{remark}\label{remark:nodes}
Consider any $X \subseteq C_1$ with $\bigl\{\{1,2\},\{1,3\},\{2,3\}\bigr\} \subseteq X$. Then there is no $x \in C_1$ with $X \subseteq S_1(x)$. Therefore $A_2$ cannot output colour $X$, and hence $X \notin C_2$.
\end{remark}

Hence graph $\cS_2$ has $|C_2| \le 55$ nodes: out of the $2^6 = 64$ candidate colours, we can exclude the empty set and $8$ other sets identified in Remark~\ref{remark:nodes}. We will now partition the remaining nodes in $16$ colour classes (independent sets).

\paragraph{Colour Classes.} 

There are four types of colour classes. First, for each $\emptyset \ne X \subseteq [3]$ we define a singleton colour class
\[
    \mathcal{X}_0(X) = \Bigl\{ \bigl\{ \{x\} : x \in X \bigr\} \Bigr\},
\]
that is, an independent set of size $1$. Then for each triple
\[
    (i,j,k) \in \bigl\{ (1,2,3),\, (1,3,2),\, (2,3,1) \}
\]
we have three colour classes:
\begin{align*}
\mathcal{X}_1(i,j,k) &= \Bigl\{ X \in C_2 : \bigl\{ \{i,j\}, \{i,k\} \bigr\} \subseteq X \subseteq  \bigl\{ \{i,j\}, \{i,k\}, \{i\}, \{j\}, \{k\} \bigr\} \Bigr\} \\
\mathcal{X}_2(i,j,k) &= \Bigl\{ X \in C_2 : \bigl\{ \{i,j\}, \{k\} \bigr\} \subseteq X \subseteq  \bigl\{ \{i,j\}, \{i\}, \{j\}, \{k\} \bigr\} \Bigr\}, \\
\mathcal{X}_3(i,j,k) &= \Bigl\{ X \in C_2 : \bigl\{ \{i,j\} \bigr\} \subseteq X \subseteq  \bigl\{ \{i,j\}, \{i\}, \{j\} \bigr\} \Bigr\}.
\end{align*}
In total, there are $7$ singleton colour classes, and $3 \times 3$ other colour classes, giving in total $16$ colour classes. Figure~\ref{fig:complement} shows the complement of a supergraph of $\cS_2$; each of the above colour classes correspond to a clique in the complement graph.

It can be verified that each of the $55$ possible nodes of $\cS_2$ is included in exactly one of the colour classes. It remains to be shown that each colour class is indeed an independent set of $\cS_2$.

The singleton classes form independent sets trivially. We handle each type of the remaining colour classes separately. Recall that there is an edge $\{X,Y\}$ in $\cS_2$ if either $X \in S_2(Y)$ or $Y \in S_2(X)$. 

\begin{figure}[p]
 \begin{center}
\includegraphics[scale=0.65,page=3]{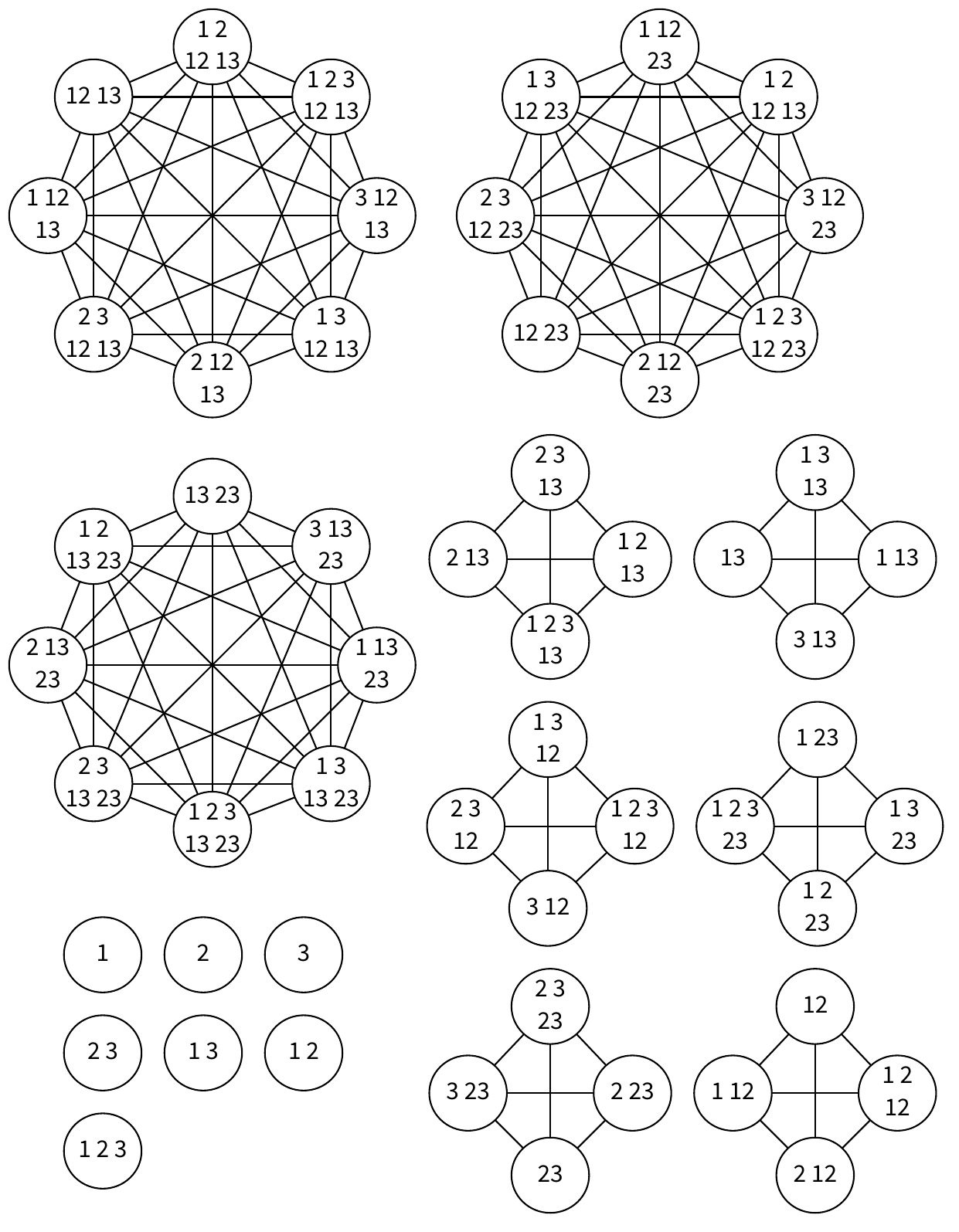}
 \end{center}
\caption{This illustrations shows the \emph{complement} of a graph we call $\cS_2^*$. For any algorithm $A$, the successor graph $\cS_2(A)$ is a subgraph of $\cS_2^*$, and hence, a proper colouring of $\cS_2^*$ is a proper colouring of $\cS_2(A)$. Each clique in the figure corresponds to a colour class in $\cS_2^*$. We use a shorthand notation: for example, the circle labelled with ``1~2~12'' is the node $\{\{1\},\{2\},\{1,2\}\}$. \label{fig:complement}}
\end{figure}

\begin{lemma}
 The class $\mathcal{X}_1(i,j,k)$ forms an independent set in $\cS_2$.
 \end{lemma}
\begin{proof}
 Let $\mathcal{X} = \mathcal{X}_1(i,j,k)$. Observe that for any $X \in \mathcal{X}$ we have $\big\{ \{i,j\}, \{i,k\} \big\} \subseteq X$ and $\{j,k\} \notin X$. From Remark~\ref{remark:successors} it easily follows that the relation $S_1$ satisfies 
\[
\big\{ \{i,j\}, \{i,k\} \big\} \subseteq X \subseteq 2^{S_1(x)} \implies x = \{j,k\}.
\]
In particular, we get that $X \in \mathcal{X} \implies X \in R_1(\{j,k\})$.

In order to show that $\mathcal{X}$ is an independent set in $\cS_2$, let $Y$ and $Z$ be vertices of $\cS_2$ such that $Y \in S_2(Z)$. First, if $Y \in \mathcal{X}$, then $Y \in R_1(\{j,k\})$. In particular, this means that $\{j,k\} \in Z$ and we get that $Z \notin \mathcal{X}$. For the second case, suppose $\{j,k\} \notin Z$. This means that a node with colour $Z$ cannot have a successor with colour $\{j,k\}$ in a colouring produced by $A_1$. Thus, it must be that $Y \notin R_1(\{j,k\})$. By the earlier observation, we get that $Y \notin \mathcal{X}$.
\end{proof}

\begin{lemma}
The class $\mathcal{X}_2(i,j,k)$ is independent in $\cS_2$. 
\end{lemma} 
\begin{proof}
Let $\mathcal{X} = \mathcal{X}_2(i,j,k)$. In this class, for every $X \in \mathcal{X}$ it holds that $\big\{ \{i,j\}, \{k\} \big\} \subseteq X$. 
From Remark~\ref{remark:successors} it follows that
\[
 \left\{ \{i,j\}, \{k\} \right\} \subseteq X \subseteq 2^{S_1(x)} \implies x \in \left\{ \{i,k\}, \{j,k\} \right\}. 
\]
Thus, $X \in \mathcal{X} \implies X \in R_1(\{i,k\}) \cup R_1(\{j,k\})$. 

In order to show that the class $\mathcal{X}$ forms an independent set in $\cS_2$, suppose $Y \in S_2(Z)$ for some $Y,Z \in C_2$. First, if $Y \in \mathcal{X}$, then we have that either $\{i,k\} \in Z$ or $\{j,k\} \in Z$ so $Z \notin \mathcal{X}$. Second, if $Z \in \mathcal{X}$, then $\big\{\{i,k\}, \{j,k\}\big\} \cap Z = \emptyset$. This means that a node with colour $Z$ cannot have successor of colour $\{i,k\}$ or $\{j,k\}$ as a successor, hence $Y \notin R_1(\{i,k\}) \cup R_1(\{j,k\})$.
\end{proof}

\begin{lemma}
 The class $\mathcal{X}_3(i,j,k)$ forms an independent set in $\cS_2$.
\end{lemma}
\begin{proof}
 Let $\mathcal{X} = \mathcal{X}_3(i,j,k)$. Observe that for any $X \in \mathcal{X}$ it holds that $\{i,j\} \in X$ and $\big\{ \{i,k\}, \{j,k\}, \{k\} \big\} \cap X = \emptyset$. Using Remark~\ref{remark:successors} we can check that relation $S_1$ satisfies
\[
 X \subseteq 2^{S_1(x)} \implies x = k.
\]
Thus, $X \in \mathcal{X} \implies X \in R_1(k)$. 

To see that $\mathcal{X}$ is an independent set in $\cS_2$, let $Y \in S_1(Z)$. There are two cases to consider. First, if $Y \in \mathcal{X}$, then $Y \in R_1(k)$. That is a node with colour $k$ can output colour $Y$. Thus, $k \in Z$, so we must have that $Z \notin \mathcal{X}$. Second, if $Z \in \mathcal{X}$, then $k \notin Z$ which means that $Y \notin R_1(k)$. Thus, $Y \notin \mathcal{X}$.
\end{proof}

\subsection{Computational Proof of Lemma~\ref{lemma:s2-16-cols}}\label{sec:computer-proof}

We now give a \emph{computational proof} of Lemma~\ref{lemma:s2-16-cols}, that is, we show how to easily verify with a computer that the claim holds. Essentially this amounts to checking that for every choice of $A=A_0$, the successor graph $\cS_2(A)$ is colourable with 16 colours. However, since any successor graph $\cS_2(A)$ depends on the choice of initial one-sided 3-colouring algorithm $A = A_0$, and there are potentially many choices for $A$, we instead bound the chromatic number of a closely-related graph $\cS_2^*$ that contains $\cS_2(A)$ for any $A$ as a subgraph. 

To construct the graph $\cS_2^*$, we consider the successor graph of a ``worst-case'' algorithm that may output ``all possible'' colours in its output set. Specifically, this means that we simply replace the subset relation in Remarks~\ref{remark:colors}~and~\ref{remark:successors} with an equality. Therefore, the graph $\cS_2^*$ can be constructed using a fairly straightforward computer program, with a mechanical application of the definitions. The end result is a dense graph on $55$ nodes; its complement is shown in Figure~\ref{fig:complement}. 

It is now easy to discover a colouring of graph $\cS_2^*$ that uses 16 colours with the help of e.g.\ modern SAT solvers. This implies that any subgraph $\cS_2(A)$ can also be coloured with $16$ colours and Lemma~\ref{lemma:s2-16-cols} follows.

\section{Main Theorems}

We now have all the pieces for proving Theorem~\ref{thm:main}:

\mainthm* 
\begin{proof}
Let $n = \pt{2k+1} + 1$ for any $k \ge 2$. Be Lemma~\ref{lemma:hardness} we have the identity
\begin{equation}
\label{eq:identity}
 C(n , 3) = \left \lceil T(n,3)/2 \right \rceil
\end{equation}
and from Theorems~\ref{thm:ub}~and~\ref{thm:lb} we get that
\[
 2k+1 \le T(n ,3) \le 2k+2,
\]
which together with (\ref{eq:identity}) yields $C(n,3) = k+1$. Since $\logst n = 2k+2$ it follows that $C(n,3) = k+1 = \logst n / 2$.
\end{proof}

For the remaining values of $n$ we get almost-tight bounds. There remains a slack of \emph{one} communication round in the upper and lower bounds for $C(n,3)$.

\begin{theorem}
 For any $n \ge 4$, 
 \[
  \left \lceil \frac{1}{2} \left( \logst n -1 \right) \right \rceil \le C(n,3) \le \left\lceil \frac{1}{2} \left( \logst n + 1 \right) \right\rceil.
 \]
\end{theorem}
\begin{proof}
 For $n=4$, we have shown that $T(4,3) = 2$ so the bounds follow. Fix $n > 4$. Now there exists some $h > 1$ such that $n \in \{\pt{h} + 1, \dots, \pt{h+1}\}$ and $h = \logst n - 1$. Theorems~\ref{thm:ub}~and~\ref{thm:lb} give us the bounds
\[
 \logst n - 1 = h \le T(n,3) \le h+2 = \logst n + 1
\]
and since $C(n,3) = \lceil T(n,3)/2 \rceil$, the claimed bounds follow.
\end{proof}

\section{Conclusions and Discussion}

In this work we gave exact and near-exact bounds on the complexity of distributed graph colouring. The key result is that the complexity of colour reduction from $n$ to $3$ on directed paths and cycles is exactly $\frac{1}{2}\logst n$ rounds for infinitely many values of $n$, and very close to it for all values of $n$.

In essence, we have shown that the colour reduction algorithm by Naor and Stockmeyer is near-optimal, while the algorithm by Cole and Vishkin is suboptimal. We have also seen that Linial's lower bound had still some room for improvements.

One of the novel techniques of this work was the use of \emph{\textbf{computers in lower-bound proofs}}. Two key elements are results of a computer search:
\begin{itemize}
    \item Lemma~\ref{lemma:comp-bounds}: The proof of $T(16,3) \ge 3$ is based on the analysis of the chromatic number of the neighbourhood graph $\mathcal{N}_{7,1}$.
    \item Lemma~\ref{lemma:16-cols}: The proof of $T(n,16) \le T(n, 3) - 2$ is based on the analysis of the chromatic number of the successor graph $\cS_2$.
\end{itemize}
In both cases we used computers to analyse the chromatic numbers of various successor graphs and neighbourhood graphs, in order to find the right parameters for our needs.

The idea of analysing \emph{\textbf{neighbourhood graphs}} and their chromatic numbers is commonly used in the context of human-designed lower-bound proofs \cite{linial92locality,naor91lower,kuhn06complexity,fraigniaud07distributed}. It is also fairly straightforward to construct neighbourhood graphs so that we can use computers and graph-colouring algorithms to discover new upper bounds~\cite{rybicki11msc}, and the same technique can be used to prove lower bounds on $T(n,c)$ for small, fixed values of $n$ and $c$; in our case we used it to bound $T(16,3)$. However, this does not yield bounds on, e.g., $T(n,3)$ for large values of $n$.

The key novelty of our work is that we can use the chromatic number of \emph{\textbf{successor graphs}} to give improved bounds on $T(n,3)$ for all values of $n$. To do that, it is sufficient to find a successor graph $\cS_k$ with a small chromatic number, apply Lemma~\ref{lemma:chromatic}. The same technique can be also used to study $T(n,c)$ for any fixed $c \ge 3$.

\section*{Acknowledgements}

We thank Juho Hirvonen for helpful comments. Parts of this work are based on the first author's MSc thesis~\cite{rybicki11msc}. Computer resources were provided by the Aalto University School of Science ``Science-IT'' project, and by the Department of Computer Science at the University of Helsinki.

\bibliographystyle{plain}
\bibliography{coloring}

\end{document}